\documentclass[runningheads]{llncs}
\usepackage[utf8]{inputenc}
\usepackage{amsmath}
\usepackage{amsfonts}
\usepackage{amssymb}
\usepackage[normalem]{ulem}
\usepackage{enumerate}
\usepackage{xcolor}
\usepackage[hidelinks]{hyperref}
\usepackage{mathrsfs}
\usepackage{graphicx}

\newcommand{\Fqm}{\mathbb{F}_{q^m}}
\newcommand{\Fqt}{\mathbb{F}_{q^t}}
\newcommand{\Fqd}{\mathbb{F}_{q^d}}
\newcommand{\Fq}{\mathbb{F}_q}
\newcommand{\F}{\mathbb{F}}
\newcommand{\cc}{\mathbf{c}}
\newcommand{\MM}{\mathbf{M}}
\newcommand{\0}{\mathbf{0}}
\renewcommand{\H}{\mathcal{H}}
\newcommand{\C}{\mathcal{C}}
\newcommand{\D}{\Delta}
\renewcommand{\S}{\Sigma}
\newcommand{\GG}{\mathbf{G}}

\newcommand{\<}{\left<}
\renewcommand{\>}{\right>}
\newcommand{\rank}{\mathsf{rank}}
\newcommand{\dr}{\mathrm{d}_r}
\renewcommand{\dh}{\mathrm{d}_H}
\renewcommand{\AA}{\mathbf{A}}
\renewcommand{\SS}{\mathbf{S}}
\newcommand{\ee}{\mathbf{e}}

\renewcommand{\aa}{\mathbf{a}}
\newcommand{\bb}{\mathbf{b}}

\newcommand{\xx}{\mathbf{x}}
\newcommand{\yy}{\mathbf{y}}

\newcommand{\fH}{\mathfrak{H}}
\renewcommand{\P}{\mathcal{P}}
\newcommand{\hw}{\mathsf{w}_H}
\newcommand{\PP}{\mathbb{P}}
\begin{document}
\title{Antipodal two-weight rank metric codes\thanks{During the course of this work, the first named author was partially supported by Grant 280731 from the Research Council of Norway and currently, she is supported by a postdoctoral researcher grant by iHub Anubhuti-IIITD Foundation (IHUB), IIIT Delhi. The second named author is partially funded by the National Science Foundation (NSF) grant CNS-1906360.}}
%
%
\author{Rakhi Pratihar\inst{1}\and
Tovohery Hajatiana Randrianarisoa\inst{2}}
\authorrunning{R. Pratihar and T. H. Randrianarisoa}
%
\institute{Indraprastha Institute of Information Technology Delhi, India
\email{pratihar.rakhi@gmail.com} \and
Florida Atlantic University, Boca Raton, Florida, USA\\
\email{tovo@aims.ac.za}
}
\maketitle              
\begin{abstract}
We consider the class of linear antipodal two-weight rank metric codes and discuss their properties and characterization in terms of $t$-spreads. It is shown that the dimension of such codes is $2$ and the minimum rank distance is at least half of the length. We construct antipodal two-weight rank metric codes from certain MRD codes. A complete classification of such codes is obtained, when the minimum rank distance is equal to half of the length. As a consequence of our construction of two-weight rank metric codes, we get some explicit two-weight Hamming metric codes.
\end{abstract}

\section{Introduction}

Let $q$ be a prime power and let $\Fq$ be a finite field of size $q$.
For a positive integer $n$, the \emph{Hamming weight} of an element $\cc = (c_1,\dots,c_n)$ in  $\Fq^n$ is defined by 
\[
\hw(\cc) = |\{ i: c_i\neq 0, 1\leq i\leq n\}|.
\]

The function ``$\hw$" induces a metric $\dh$ on
$\Fq^n$ where $\dh(\cc,\cc') = \hw(\cc - \cc')$ for $\cc,\cc'$ in $\Fq^n$.  
An $[n,k,d]$ (linear) Hamming metric code $\C$ over $\Fq$ is an $\Fq$-linear subspace of $\Fq^n$ of dimension $k$ such that $\min_{\cc\in \C\backslash \{\0\}} \hw(\cc) =d$.
A two-weight linear Hamming metric code over $\Fq$ is linear code such that the weight distribution (i.e. the set of possible weights) of $\C$ is of the form $\{0,d,d'\}$ such that $0<d<d'$. 
The terminology is explained by the fact that there are only two possible non-zero weights for the codewords.
An $[n,k,d]$ linear Hamming metric code over $\Fq$ is called \emph{antipodal two-weight} or ATW if $d < n$ and any nonzero codewords have Hamming weight either $d$ or $n$ and there is at least one codeword of Hamming weight $n$.

Linear codes with few distinct weights are important in coding theory, both from practical and theoretical point of view. These classes of codes are for example used in secret sharing schemes  \cite{Massey93,YD02} and authentication schemes \cite{DW05}.
For codes with Hamming metric, the\emph{ constant weight codes} i.e. codes where all the nonzero codewords have the same weight and the \emph{two-weight} linear codes have been studied extensively. 
In \cite{Bon83}, Bonisoli gives a characterization of the constant weight codes.
The class of two-weight linear Hamming metric codes has been investigated by Delsarte \cite{Del78}; see \cite{CK86} for a systematic exposition. A classification of the subclass of ATW linear codes with Hamming metric is obtained in \cite{JT18}.

We are interested in the $q$-analogues of some of these results. For linear rank metric codes, the class of constant rank weight codes are completely classified in \cite{Ran20} following a geometric approach. It is proved that, up to equivalence, there is only one non-degenerate constant rank weight code if the code has dimension at least $2$. 
In this article we consider the class of ATW rank metric codes.  Using the classification of constant rank weight codes, we show that the dimension of such codes is $2$ and the minimum distance $d$ must be at least $\frac{n}{2}$. We provide an equivalent characterization of ATW rank metric codes in terms of $t$-spreads, more precisely, as $t$-subspreads of Desarguesian $t$-spreads induced by $q$-systems associated to ATW rank metric codes. It is proved that, up to equivalence, there exists only one ATW code for the case $d=\frac{n}{2}$ and a complete classification of such codes is also obtained. For the case of $d > \frac{n}{2}$, we get partial results regarding the classification. We construct ATW rank metric codes using MRD codes of suitable parameters. 

The article is organized as follows. In the next section, we collect some preliminaries about notions such as rank metric codes and $t$-spreads. In Section \ref{sec:3}, we derive some properties of the ATW rank metric codes and $t$-spreads and establish a relation between these two objects.  Section \ref{sec:4} deals with classification and construction of ATW rank metric codes. In Section \ref{sec:5}, we provide a construction of two-weight Hamming metric codes from two-weight rank metric codes. Finally, we conclude the paper with a few open questions.

\section{Preliminaries}\label{sec:2}
Throughout, we use $\aa \cdot \bb$ to denote the usual dot product of two vectors $\aa, \bb\in\Fqm^n$. For $\aa=(a_1,\dots,a_n)\in \Fqm^n$, we denote by $\<a_1,\dots,a_n\>_{\Fq}$ the $\Fq$-subspace of $\Fqm$ generated by $\{a_1,\dots,a_n\}$.

\subsection{Rank metric codes}\label{sub:2.1}
Let $\Fqm/\Fq$ be a finite field extension of degree $m$. We endow the space $\Fqm^n$ with a norm, called $\rank$, where $\rank(\cc) = \dim_{\Fq}\<c_1,\dots,c_n\>_{\Fq}$ for any $\cc = (c_1,\dots,c_n)$ in  $\Fqm^n$. This norm induces a metric on $\Fqm^n$, which is defined by 

\[
\dr(\cc_1,\cc_2)=\rank(\cc_1-\cc_2), \text{ for } \cc_1, \cc_2 \in \Fqm^n.
\]

\begin{definition}[Rank metric codes]
Let $\Fqm/\Fq$ be a field extension of degree $m$. An $[n,k,d]$ rank metric code $\C$ over $\Fqm/\Fq$ is a $k$-dimensional $\Fqm$-subspace of $\Fqm^n$ such that $d=\min \{\rank(\cc) \colon \cc\in \C\backslash \{\0\}\}$. Here $d$ is called the minimum (rank) distance of the code.
\end{definition}

When we are not concerned with the minimum distance, we simply call the codes as $[n,k]$ rank metric codes. In \cite{Del78}, it was shown that for any $[n,k,d]$ rank metric code over $\Fqm/\Fq$, $d\leq n-k+1$. Similar to the Hamming metric case, this inequality is called the Singleton bound and when the bound is attained i.e. $d=n-k+1$, the code is called a \emph{maximum rank distance} (MRD) code. The first construction of MRD codes was independently presented in \cite{Del78,Gab85}. Later, many constructions of new classes of MRD codes were obtained, e.g., in \cite{GK05,PR22,She16} to name a few.

For the purpose of classifying the rank metric codes, we need to have a notion of equivalence of two codes. There are various notions of equivalence of two rank metric codes. In this paper we consider the following notion defined in \cite{Mor14}.
\begin{definition}\cite{Mor14}
Two rank metric codes $\C$ and $\C'$ over $\Fqm/\Fq$ are called equivalent if $\C=\alpha \C' \MM$ for some $\alpha\in \Fqm^{\times}$ and $\MM\in \Fq^{n\times n}$ invertible, where $\C'\MM=\{\cc\MM\colon \cc\in \C'\}$. It is equivalent to saying that any two  generator matrices $\GG$ and $\GG'$, corresponding to $\C$ and $\C'$, respectively, satisfy the relation $\GG=\SS \GG'\MM$, where $\SS\in \Fqm^{k\times k}$ is invertible.
\end{definition}

A rank metric code $\C$ is called \emph{non-degenerate} if the columns of its generator matrix $\GG$ are linearly independent over $\Fq$. Otherwise, $\C$ is called \emph{degenerate} and in this case, $\C$ is equivalent to a code $\{(\cc|\0)\colon \cc\in \C'\}$, where $\C'$ is non-degenerate of shorter length. Throughout this paper, our code is assumed to be non-degenerate unless otherwise specified. 

In \cite{Ran20}, rank metric codes are described in terms of the geometric object called $q$-systems. This geometric description offers the advantage of studying the distance of a code by looking at the hyperplane sections of a corresponding $q$-system. Next we review some of notions related to the geometric description of rank metric codes relevant for our study.

\begin{definition}[$q$-system]
Let $k,n$ be positive integers such that $k \leq n$. An $[n,k]$ $q$-system $X$ in $\Fqm^k$ is an $\Fq$-subspace of $\Fqm^k$ of dimension $n$. If the set $\{\cc_i \in \Fqm^k \colon 1\leq i\leq n\}$ forms an $\Fq$-basis of $X$, then we write $X=\<\cc_1,\dots,\cc_n\>_{\Fq}$.
\end{definition}

Let $\GG$ be a generator matrix of a non-degenerate $[n,k]$ rank metric code $\C$ over $\Fqm/\Fq$. The $\Fq$-space generated by the columns of $\GG$ is an $[n,k]$ $q$-system. As it was shown in \cite[Theorem 2]{Ran20}, this construction gives a one-to-one correspondence between equivalence classes of $q$-systems and rank metric codes. The key point of this approach is the geometric interpretation of rank of a codeword which we explain below. 

\begin{lemma}\label{lem:1}
For a vector $\cc\in \Fqm^n$, $\rank(\cc)=n-\dim_{\Fq} \mathrm{Ker}\,\phi_\cc$ where $\phi_\cc$ is the $\Fq$-linear map defined as
    \begin{align*}
         \phi_\cc\colon \Fq^n &\longrightarrow \Fqm\\
         \aa &\longmapsto \cc\cdot\aa.
    \end{align*}
\end{lemma}
    
\begin{proof}
It is a straightforward consequence of the rank-nullity theorem.
\qed\end{proof}

For any $\xx \in \Fqm^k\backslash \0$, we consider the following $\Fqm$-linear map 

\begin{align}\label{eq:psi}
    \psi_{\xx}\colon\Fqm^k &\longrightarrow \Fqm\\
    \nonumber\ee & \longmapsto \xx\cdot\ee.
\end{align}

It is easy to see that kernel of $\psi_{\xx}$ defines a unique hyperplane, which we name as $H_{\xx}$. In fact, any hyperplane of $\Fqm^k$ can be seen as kernel of $\phi_{\xx}$ for some $\xx \in \Fqm^k$. 

Now let $\C$ be an $[n,k]$ rank metric code over $\Fqm/\Fq$. Any codeword $\cc \in \C$ is of the form $\cc=\xx_\cc\GG$ for some $\xx_\cc\in \Fqm^k$. Then $\cc$ defines a unique hyperplane $H_\cc = \mathrm{Ker}~ (\psi_{\xx_{\cc}})$ of $\Fqm^k$.
Conversely, from the discussion above, a hyperplane $H$ of $\Fqm^k$ corresponds to a codeword $\cc=\xx\GG$ defined by a vector $\xx\in \Fqm^k$. 

\begin{lemma}\label{lem:2}
Let $\C$ be an $[n,k]$ (non-degenerate) rank metric code over $\Fqm/\Fq$ and let $X$ be the $[n,k]$ $q$-system corresponding to a generator matrix $\GG$ of $\C$. Then for any $\cc = {\xx_{\cc}\GG} \in \C$,
\begin{equation}\label{eq:rank}
\rank(\cc)=n-\dim_{\Fq} \left( X\cap H_{\cc} \right), \text{ where}
\end{equation}
$H_{\cc}=\mathrm{Ker}~ \psi_{\xx_{\cc}}$ is considered as an $\Fq$-subspace of $\Fqm^k$ of dimension $m(k-1)$.
\end{lemma}
\begin{proof}
Following the definition of $\psi_{\xx_{\cc}}$ in Equation \eqref{eq:psi}, for any $\ee\in \left( X\cap H_{\cc} \right)$, we have $\xx_\cc\cdot \ee=0$. Moreover, $X$ being the $q$-system corresponding to $\GG$, we get $\ee=\GG \aa^T$ for some $\aa\in \Fq^n$. Note that $\aa$ is unique since $\C$ is non-degenerate. This implies that $\cc\cdot \aa = 0$. Therefore, $\ee\in X \cap H_{\cc}$ corresponds to a unique element $\aa \in \mathrm{Ker} ~ \phi_{\cc}$.  Also, the same argument backward shows that any element of $ \mathrm{Ker}\, \phi_\cc$ corresponds to a unique element of $X\cap H_\cc$. Hence $|X\cap H_\cc|= |\mathrm{Ker}\, \phi_\cc|$ and the result follows from Lemma \ref{lem:1}.
\qed\end{proof}

The geometric approach offers the advantage that if the $q$-system $X$ corresponding to the code $\C$ has nice properties, then it can be simpler to study the intersection of $X$ with the hyperplanes of $\Fqm^k$. That makes it easier to study the weights of the codewords of $\C$. Furthermore, if the hyperlanes $\H_\cc$ are replaced with subspaces of arbitrary dimension, then it defines the ``\emph{rank weight}'' of subcodes of $\C$. And these can be used to define the ``generalized rank weights'' of a rank metric code. For more details on the geometric approach, one can look at \cite{Ran20}.
Using the geometric approach, the complete classification of the constant weight rank metric codes has been obtained and we recall the result below.

\begin{definition}\cite[Definition 11]{Ran20}
Let $X$ be the $\Fq$-vector space $\Fqm^k$ of dimension $mk$. A Hadamard rank metric code $H_1(q,m,k)$ is an $[mk,k,m]$ linear code over $\Fqm/\Fq$ with a generator matrix where the columns consist of an $\Fq$-basis of $X=\Fqm^k$.
\end{definition}

Notice that Hadamard rank metric codes are unique up to equivalence. The following theorem is the complete characterization of (linear) constant weight rank metric codes i.e. all non-zero codewords have the same rank $d$ equal to the minimum distance of the code.

\begin{theorem} \cite[Theorem 12]{Ran20}\label{thm:1}
Let $\C$ be an $[n, k, d]$ non-degenerate constant weight linear rank metric code over $\Fqm/\Fq$. We have the following two cases:
\begin{enumerate}[(a)]
\item  If $k = 1$, then $\C = \<(a_1,\dots, a_n)\>_{\Fqm}$ with $\rank(a_1,\dots, a_n) = d = n$, and
\item  if $k>1$, then $\C$ is an $[mk,k,m]$ Hadamard rank metric code $H_1(q,m,k)$.
\end{enumerate}
\end{theorem}

Theorem \ref{thm:1} gives a complete description of non-degenerate (and thus degenerate) $[n,k,d]$ rank metric codes whose all non-zero codewords have the same rank, that is equal to $d$. The next logical question is to consider codes where there are two possible ranks for the non-zero codewords.

\begin{definition}[Two-weight rank metric codes]\label{two-weight}
An $[n,k,d]$ rank metric code $\C$ over $\Fqm/\Fq$ is called a two-weight rank metric code if for any $\cc\in \C$, $\rank(\cc)\in \{0,d,d'\}$ where $0 < d < d'$.
\end{definition}

It is shown in \cite[Proposition 3.11]{ABNR21} that, for an $[n,k,d]$ rank metric code over $\Fqm/\Fq$, there exists a codeword with rank equal to $\min\{n,m\}$. Therefore, for a two-weight rank metric code, there are two possible weight distributions: $\{0,d,m\}$ or $\{0,d,n\}$. Now we consider the particular class of two-weight rank metric codes that have a weight distribution of the latter type i.e. $\{0, d, n \}$.

\begin{definition}[Antipodal two-weight (ATW) rank metric codes]
An $[n,k,d]$ linear rank metric code over $\Fqm/\Fq$ with $d <n$ is called an \emph{antipodal two-weight} or ATW rank metric code if any nonzero codeword has rank either $d$ or $n$ and there is at least one codeword of full rank $n$. 
\end{definition}

The most straightforward examples of antipodal two-weight rank metric codes are given by MRD codes of dimension $2$. To see this, let $\C$ be an $[n,2,n-1]$ MRD codes over $\Fqm/\Fq$, where $n\leq m$. Then, for any $\cc\in \C$ such that $\cc\neq \0$, it is clear that $n-1\leq \rank (\cc)\leq n$. Hence, the two dimensional MRD codes are $[n,2,n-1]$ ATW rank metric codes. As, we explained earlier, there are many constructions of MRD codes \cite{Gab85,GK05,PR22,She16} and those already provide inequivalent classes of ATW rank metric codes of dimension $2$.

\subsection{\texorpdfstring{$t$}{t}-Spreads}\label{sub:2.2}

Spreads are widely studied objects in finite geometry \cite{BJJ01,NVB07,Hir79}. There are several works which connect them to coding theory. For example, spreads when considered as subspace codes are optimal subspace codes \cite{MGR08} i.e. they are the largest given their parameters. In this work, we give a connection between spreads and ATW rank metric codes. We start with a brief review of the basic notions and relevant results related to spreads.

\begin{definition}[$t$-Spreads]\label{def:spread}
A $t$-spread is a pair $(V,\S)$ where $V$ is a vector space of dimension $n$ over $\Fq$ and $\S$ is a set of subspaces of $V$ of dimension $t$ such that $\bigcup_{S\in \S} S=V$ and for all $S_1\neq S_2\in \S$, $S_1\cap S_2=\{\0\}$.
If the ambient space $V$ is clear from the context, we simply write $\S$ to denote the $t$-spread. 
\end{definition}

 From now onwards, whenever we write ``spread'', we mean that $(V,\S)$ is a $t$-spread where $t=(\dim_{\Fq} V)/2$. It is well-known \cite{Seg64} that an $\Fq$-space $V$ of dimension $n$ has a $t$-spread if and only if $t$ divides $n$. Moreover, $|\Sigma| = \frac{q^n-1}{q^t-1}$.

Two $t$-spreads $\Sigma_1,~ \Sigma_2$ of $V$ are called $\emph{equivalent}$ if there exists a collineation $\alpha$ of $\Gamma L (V)$ such that $\Sigma_1^{\alpha} = \Sigma_2$.

Suppose $n=lt$. Any $r$-dimensional $\Fqt$-subspace of $\Fqt^l$ is an $rt$-dimensional $\Fq$-subspace of $\Fqt^l$. Now, let $\H$ be the set of all one dimensional subspaces of $\Fqt^l$ and thus $\#\H = \frac{q^n-1}{q^t-1}$. These $1$-dimensional $\Fqt$-subspaces are considered as $t$-dimensional $\Fq$-subspaces of $\Fqt^l\simeq \Fq^{tl}=\Fq^n$ and they intersect trivially. Hence, up to isomorphism, $\H$ is a set of $\frac{q^n-1}{q^t-1}$ non-intersecting $t$-dimensional subspaces of $\Fq^{n}$. This gives the first example of $t$-spreads.

\begin{definition}[Desarguesian $t$-spread]\label{def:desarguesian}
Let $\mathcal{D}_{l,t,q}$ be the $t$-spread of $\Fqt^{l}$ consisting of the $t$-dimensional subspaces of $\Fqt^{l}$ defined by the $1$-dimensional $\Fqt$-subspaces of $\Fqt^l$. A $t$-spread $\Sigma$ is called Desarguesian if $\Sigma$ is equivalent to $\mathcal{D}_{l,t,q}$ for some $l, ~t$ and $q$.
\end{definition}

\begin{definition}\label{def:6}
Let $V$ be an $\Fq$-vector space. Let $(V,\S)$ be a $t$-spread in $V$. For a subspace $W$ of $V$, the projection of $\S$ onto $W$, denoted as $\S_{W}$, is defined as follows
\[
\S_{W}=\{ S\cap W\colon S\in \S\} \backslash \{\{\0\}\}.
\]
The pair $(W,\S_{W})$ is called a subsystem of $(V,\S)$. 
\end{definition}

\begin{remark}
When the intersection is equal to the zero space $\{\0\}$, we remove it from the subsystem. For any arbitrary subspace $W\subset V$, the projection $\S_W$ is not necessarily a $t'$-spread, because the elements of $\S_W$ may not have the same dimension. 
\end{remark}

\begin{definition}\label{def:subspread}
Let $(V,\S)$ be a $t$-spread in $V$ and let $(W,\S_{W})$ be a subsystem of $(V,\S)$. If $(W,\S_{W})$ is itself $t'$-spread for some $t'\leq t$, then it is called a $t'$-subspread induced by $(V,\S)$ on $W$. 
If a $t'$-subspread is a spread, then we simply call it a subspread.
\end{definition}

Thus the notion of $t$-subspreads generalizes the notion of subspreads of \cite[Definition 4.4.1]{BJJ01} to arbitrary $t$.
Later in Theorem \ref{pro:spread}, we will show how $t$-subspreads are related to the main subject of this paper i.e. antipodal two-weight rank metric codes.

\section{Antipodal two-weight rank metric codes and \texorpdfstring{$t$}{t}-spreads}\label{sec:3}
In this section, we discuss properties of the \emph{antipodal} two-weight rank metric codes and $t$-spreads and establish a relation between these two objects.

\subsection{Antipodal two-weight (ATW) rank metric codes}\label{sub:3.1}
For an $[n,k,d]$ ATW rank metric code over $\Fqm/\Fq$, we must have $n\leq m$ because it has a codeword of full rank $n$. We first give a simple description of the form of the generator matrix of an ATW rank metric code. As we progress, we will refine this form.

\begin{lemma}\label{lem:3}
Let $\C$ be a non-degenerate $[n,k,d]$ ATW rank metric code over $\Fqm/\Fq$. Then $\C$ is equivalent to a rank metric code with a generator matrix of the form
\[
\GG=\begin{pmatrix}
\cc_1 & \cc_2 \\
\AA & \0
\end{pmatrix}, \text{ where}
\]
\begin{enumerate}[(i)]
    \item\label{i} $\cc_1\in \Fqm^{n-r}$, $\cc_2\in \Fqm^{r}$ for some positive integer $r$ and $\rank(\cc_1|\cc_2)=n$.
    \item\label{iv} $\AA\in \Fqm^{(k-1)\times (n-r)}$ is a generator matrix of a non-degenerate constant weight $[n-r,k-1]$ rank metric code.
\end{enumerate}
\end{lemma}
\begin{proof}
 Since $\C$ is antipodal, there is a codeword $\cc$ of full rank $n$. After row reductions, we can have a generator matrix of the following form
\[
\GG'=\begin{pmatrix}
\cc' & \alpha \\
\AA' & \0_{(k-1)\times 1}
\end{pmatrix},
\]
where $\cc'\in \Fqm^{n-1}, \; \alpha \in \Fqm^{\times}$ such that $\cc = (\cc'|\alpha)$ and $\AA'\in \Fqm^{(k-1)\times (n-1)}$. Since the subcode $\C'\subset \C$ generated by $\left[\AA'|\0_{(k-1)\times 1}\right]$ can't have a codeword of rank $n$, then it must be a constant rank weight code with minimum distance $d$. In case $\AA'$ generates a degenerate rank metric code, then we can find some invertible matrix $\MM\in \Fq^{n\times n}$ such that 
\[
\left[\AA'|\0_{(k-1)\times 1}\right]\MM = \left[\AA|\0_{(k-1)\times r}\right],
\]
where $1\leq r\leq n-1$ and $\AA\in  \Fqm^{(k-1)\times (n-r)}$  is a generator matrix of a non-degenerate constant weight rank metric code. Therefore, $\GG:=\GG'\MM$ has the desired form and the codeword $(\cc_1|\cc_2) = (\cc'|\alpha)\MM$ has rank $n$ since $M$ is invertible.
\qed\end{proof}

It is clear that any two-weight rank metric code cannot be of dimension $1$. Such codes are constant weight rank metric codes. 
Now, following Theorem \ref{thm:1}, we can say that the matrix $\AA$ in the decomposition of $G$ in Lemma \ref{lem:3} generates either a code of dimension $1$ or a Hadamard rank metric code. And the later case leads to the following result.

\begin{theorem}\label{thm:2}
There is no $[n,k,d]$ ATW rank metric code over $\Fqm/\Fq$ of dimension $k\geq 3$.
\end{theorem}
\begin{proof}
Let $\C$ be an $[n,k,d]$ ATW rank metric code with of dimension $k\geq 3$ and generator matrix $\GG$. 
By Lemma \ref{lem:3}, we can assume that
\[
\GG=\begin{pmatrix}
\cc_1 & \cc_2 \\
\AA & \0
\end{pmatrix},
\]
where $\AA$ is a generator matrix of a non-degenerate constant weight rank metric code of dimension $k-1 \ge 2$ and minimum distance $d$. Following Theorem \ref{thm:1}, $\AA$ is a generator matrix of a Hadamard code $H_1(q,m,k-1)$. But that implies $m=d<n \leq m$ which is a contradiction.  
\qed\end{proof}

So any ATW rank metric codes are only of dimension $2$. The next lemma gives a restriction on the distance of the code.

\begin{lemma}\label{lem:4}
For any $[n,2,d]$ ATW rank metric code, the minimum distance $d$ is at least $n/2$. 
\end{lemma}

\begin{proof}
Let $\C$ be an $[n,2,d]$ ATW rank metric code. Following Lemma \ref{lem:3}, $\C$ is equivalent to a code $\C'$ with a generator matrix $
\GG=\begin{pmatrix}
\cc_1 & \cc_2 \\
\cc_3 & \0
\end{pmatrix}$, 
where $\cc_1,\cc_3\in \Fqm^d$, $\cc_2\in \Fqm^{n-d}$, $\rank(\cc_3)=d$ and $\rank(\cc_1|\cc_2)=n$. 
After row reduction, $\GG$ can be transformed into
$
\GG'=\begin{pmatrix}
\cc_1' & \cc_2 \\
\cc_3 & \0
\end{pmatrix}
$,
where the first entry of $\cc_1'$ is zero and the two rows of the new generator matrix $\GG'$ have rank $d$. The fact that $\rank(\cc_1 | \cc_2) = n$ implies that $\cc_2$ has full rank $n-d$. Since $\rank(\cc_2)\leq\rank(\cc_1' | \cc_2)= d$, then $d\geq n-d$ i.e. $d \geq n/2$. 
\qed\end{proof}

The following example gives an ATW rank metric code with $d=n/2$. Later we show that any codes of the same parameters must be equivalent to this one.
\begin{example}\label{exa:1}
  Let $\Fqd=\left<\alpha_1,\dots,\alpha_d\right>_{\Fq}$ be an extension of degree $d$ over $\Fq$. Then, the matrix 
\[
\GG=\begin{pmatrix}
0 & 0  & \cdots & 0 & \alpha_1 & \alpha_2 & \cdots & \alpha_d \\
\alpha_1 & \alpha_2 & \cdots & \alpha_d & 0 & 0  & \cdots & 0
\end{pmatrix}
\]
generates a non-degenerate ATW $[n=2d,2,d]$ rank metric code $\C$ over $\mathbb{F}/\Fq$, where $\mathbb{F}$ is any proper extension of $\Fqd$. For the proof, we know that all codewords of $\C$ are of the form $(a_1,a_2)\GG$, where $a_1, a_2 \in \F$. Now, for any such $(a_1,a_2) \in \F^2$, either $a_2 a_1^{-1}\in \Fq$ or $a_2 a_1^{-1}\notin \Fq$. For the first case, the codeword is of rank $d$, whereas the second case produces a codeword of rank $n$. Note that we cannot take $\mathbb{F} = \Fqd$, because then, there will be no codeword of full rank in $\C$.
\end{example}

In fact this construction can be modified to provide constructions of two-weight rank metric codes which are not necessarily antipodal. As we mentioned in a paragraph after Definition \ref{two-weight}, for non-antipodal two-weight rank metric codes, the weight distribution is of the form $\{0,d,m\}$ where $m<n$. The following example gives a construction of non-antipodal two-weight rank metric codes.

\begin{example}\label{exa:2}
Let $d,n$ be positive integers such that $2d<n$. We have a sequence of finite fields $\Fq\subset \F_{q^d} \subset \F_{q^{2d}}$. Let $\{\alpha_1, \dots, \alpha_d\}$ be an $\Fq$-basis of $\F_{q^{d}}$. Then consider the rank metric code $\C$ over $\F_{q^{2d}}/\Fq$, generated by the following $k\times dk$ matrix.
  \[
\GG= \left[ 
\begin{array}{*{20}c}
\alpha_1 & \hdots & \alpha_d & 0  & \hdots & 0 & \hdots & \hdots & \hdots & 0 & \cdots & 0 \\
0 & \hdots & 0 & \alpha_1  & \hdots & \alpha_d & \hdots & \hdots & \hdots & 0 & \hdots & 0 \\
\vdots & \vdots & \vdots & \vdots  & \vdots & \vdots & \ddots & \ddots & \ddots & \vdots & \vdots & \vdots \\
0 & \hdots & 0 & 0  & \hdots & 0 & \hdots & \hdots & \hdots & \alpha_1 & \hdots & \alpha_d 
\end{array}
\right].
\]
The $i$-th row of the matrix $\GG$ is chosen in such a way that the i-th block (with $d$ entries) consists of $(\alpha_1, \dots, \alpha_d)$ and the other blocks are made of zero vectors. Therefore, $\C$ has length $n=kd$ and dimension $k >2$. Now, let us check the weights of the codewords of $\C$. Suppose that $\cc=\xx\GG$ where $\xx=(x_1,\dots,x_k)\in \F_{q^{2d}}^k$. There are two possibilities:
\begin{enumerate}[(i)]
    \item $\<x_1,\dots,x_k\>_{\F_{q^d}} = \<\beta\>_{\F_{q^d}}$ for some $\beta\in \F_{q^{2d}}$.
    \item $\dim_{\F_{q^{d}}}\<x_1,\dots,x_k\>_{\F_{q^d}}>1$.
\end{enumerate}
In the first case, $\rank\, \cc  = d$ if $\beta\neq 0$ and $\rank\, \cc = 0$ if $\beta = 0$. For the second case, without loss of generality, suppose that $\dim_{\F_{q^{d}}}\<x_1,x_2\>_{\F_{q^d}}=2$. Therefore, $\rank(x_1\alpha_1, \dots, x_1\alpha_d,x_2\alpha_1, \dots, x_2\alpha_d)=2d$. Thus
\[
2d=\rank(x_1\alpha_1, \dots, x_1\alpha_d,x_2\alpha_1, \dots, x_2\alpha_d)\leq \rank\, \cc \leq 2d. 
\]
The last inequality comes from the fact that the code is over $\F_{q^{2d}}$ so that the rank of any codewords is at most $2d$. Therefore $\rank\, \cc=2d$. 

Therefore, the code $\C$ generated by $\GG$ is a $[kd,k,d]$ non-antipodal two-weight rank metric code over $\F_{q^{2d}}/\Fq$ with weight distribution $\{0,d,2d\}$.
\end{example}

\subsection{Some properties of \texorpdfstring{$t$}{t}-spreads}\label{sub:3.2}
Here we prove some properties of $t$-spreads that will be useful in latter sections.
\begin{lemma}
Let $(V, \S)$ be a $t$-spread of $V$ where $V$ is a vector space over $\Fq$ of dimension $n=tl$. For any integer $1<r<l$, if $\{S_1,S_2,\dots,S_r\}$ are distinct subspaces in $\S$ such that $\dim_{\Fq} S_1+\dots+S_r=tr$, then there is an element $S_{r+1}$ in $\S$ such that $\dim_{\Fq} S_1+\dots+S_{r+1}=t(r+1)$.
\end{lemma}
\begin{proof}
  By contradiction, suppose that for any $S_{j}\neq S_i$, $i=1,\dots,r$, $\dim_{\Fq} S_1+\dots+S_{r}+S_j<t(r+1)$. Hence $S_{j}\cap (S_1+\dots+S_r)$ contains a non-zero vector.
  Since the elements of a $t$-spread pairwise intersect trivially, each space in $\Sigma\backslash \{S_1,\dots,S_r\}$ contains a distinct non-zero vector in $S_1+\dots + S_r$. Therefore, $\# (S_1+\dots+S_r)\backslash \{\0\} \geq \frac{q^n -1}{q^t-1}-r$. 
Hence,
\begin{align*}
    & q^{tr}-1 \geq \frac{q^n-1}{q^t-1} - r\\
    \Rightarrow & q^{t(r+1)}-q^{tr}  \geq q^n-1 - (r-1)(q^t-1)\\
    \Rightarrow & q^{n}-q^{tr} \geq q^n-1 - (r-1)(q^t-1) ~(\text{As } q^n\geq q^{t(r+1)}, \text{for } r+1\leq l)\\
    \Rightarrow & q^{tr}-1 \leq (r-1)(q^t-1) < (r-1)q^t\\
    \Rightarrow & q^{tr}-1 < rq^t-q^t +(q^t-1)\\
    \Rightarrow & q^{tr}-1 < rq^t-1.
\end{align*}
    This implies that  $q^{tr} < rq^t$. Since $t\geq 1$ and $r\geq 2$, this leads to a contradiction as $q^{tr} \geq rq^t$ and hence the Lemma is proved.
\qed\end{proof}

\begin{theorem}\label{thm:split}
Suppose that $V$ is a vector space of dimension $n$ over $\Fq$ and let $n=tl$. Let $\S$ be a $t$-spread of $V$. Then there are $S_1,\dots,S_l$ in $\S$ such that $V=S_1\oplus S_2\oplus \dots\oplus S_l$.
\end{theorem}
\begin{proof}
Choose $S_1$ and $S_2$ as any elements of $\S$. By the definition of $t$-spreads, $S_1+S_2$ is a direct sum. The previous lemma says that we can increase it to $S_1\oplus S_2\oplus S_3$ and applying the lemma repetitively, we get the result.
\qed\end{proof}

\begin{remark}
The above property can be easily seen when the spread is Desarguesian. We can choose the $S_i$'s to be the $\Fq$-subspaces defined by the elements of a basis of $\Fqt^l$ over $\Fqt$.
\end{remark}

\subsection{Relation between ATW codes and \texorpdfstring{$t$}{t}-spreads}\label{sub:3.3}

We now give a relation between $(n-d)$-spreads and $[n,2,d]$ ATW rank metric codes for $d\geq n/2$. The following lemma is obtained by using the same method as used in classifying constant rank weight codes in \cite{Ran20}.
\begin{lemma}\label{lem:hyperpartition}
Let $\C$ be an $[n,2,d]$ ATW rank metric code over $\Fqm/\Fq$ with a generator matrix $\GG$ and let $X$ be the $q$-system corresponding to $\GG$. If $\H$ is the set of all hyperplanes of $\Fqm^2$, then we can partition $\H$ into $\H_1 \sqcup \H_2$, where
\begin{equation} \label{eq:hyp1}
   \H_1 = \{ H\in \H\colon \dim_{\Fq} H\cap X = n-d \}, \text{ and }
\H_2 = \{ H\in \H\colon \dim_{\Fq} H\cap X = 0 \}. 
\end{equation}
\end{lemma}
 \begin{proof}
Let $\cc \in \C$ be a nonzero codeword and let $H_{\cc}$ be the corresponding hyperplane i.e. $H_{\cc} = \mathrm{Ker} ~\psi_{\xx_\cc}$ in Lemma \eqref{eq:psi}. Then Lemma \ref{lem:2} implies that $\Fq$-dimension of $H_{\cc} \cap X$ is $n-d$ if $\rank(\cc) = d$ and $0$ otherwise. Since any hyperplane of $\Fqm^2$ corresponds to some codewords, we get the desired partition of $\H$ by considering $\H_1$ (resp. $\H_2$) to be the set of hyperplanes corresponding to codewords of rank $d$ (resp. $n$).
 \qed\end{proof}

\begin{theorem}\label{pro:spread}
Let $\C$ be an $[n,2,d]$ ATW rank metric code over $\Fqm/\Fq$ with a generator matrix $\GG$ and let $X$ be the $q$-system corresponding to $\GG$. Let $\H$ be the set of all hyperplanes of $\Fqm^2$ and $\H_1 = \{ H\in \H\colon \dim_{\Fq} H\cap X = n-d \}$. Then $|\H_1| = \frac{q^n-1}{q^{n-d}-1}$ and thus the set $\S = \{H\cap X : H \in \H_1\}$ defines an $(n-d)$-spread $(X,\S)$ of $X$.
\end{theorem}

\begin{proof}
 Lemma \ref{lem:hyperpartition} implies $ \H =\H_1 \sqcup \H_2$ where $\H_2$ as in equation \eqref{eq:hyp1}. Note that the hyperplanes have pairwise trivial intersection. Therefore we can partition the set $\Fqm^2\backslash\{\0\}$ as
\begin{equation}\label{eq:union}
\Fqm^2 \backslash\{\0\} = \left(\bigsqcup_{H\in \H_1} H\backslash\{\0\}\right) \bigsqcup \left(\bigsqcup_{H\in \H_2} H\backslash\{\0\}\right),
\end{equation}
where $H\backslash \{\0\}$ is the set of non-zero elements of $H$.
Define a valuation $v$ on $\Fqm^2$ by
\[
v(\aa)=
\begin{cases}
1 ~\text{ if } \aa\in X, \\
0 ~\text{ otherwise},
\end{cases}
\]
and by an abuse of notation, we also define $v(S)=\sum_{\aa\in S}v(\aa)$ for any subset $S$ of $\Fqm^2$.
Since all the unions in equation \eqref{eq:union} are disjoint, then
\[
v(\Fqm^2 \backslash\{\0\}) = \sum_{H\in \H_1}v\left( H\backslash\{\0\}\right) + \sum_{H\in \H_2}v\left( H\backslash\{\0\}\right),
\]
Now $v(H\backslash\{\0\})=q^{n-d}-1$ for $H\in \H_1$, and $v(H\backslash\{\0\})=0$ if $H\in \H_2$. Furthermore $v(\Fqm^2 \backslash\{\0\}) = v(X\backslash\{\0\})$. Therefore,

\begin{equation}\label{eq:H1}
q^n-1 = |\H_1|(q^{n-d}-1).
\end{equation}

The pairwise trivial intersection of elements of $\S$ implies that $\S$ is an $(n-d)$-spread.
\qed\end{proof}

\begin{corollary}\label{cor:div}
If $\C$ is an $[n,2,d]$ ATW rank metric code, then $n-d$ divides $n$.
\end{corollary}

\begin{proof}
Theorem \ref{pro:spread} implies that a $q$-system $X$ has an $(n-d)$-spread $\S$ such that $|\S| = \frac{q^n-1}{q^{n-d}-1}$ and thus $n-d$ divides $n$.
\qed\end{proof}

Corollary \ref{cor:div} recovers the result of Lemma \ref{lem:4}.

\begin{corollary}
The weight distribution of an $[n,2,d]$ ATW code over $\Fqm$ is 
\[
\left(1, \, (q^m-1)\frac{q^n-1}{q^{n-d}-1}, \, (q^{2m}-1) - (q^m-1)\frac{q^n-1}{q^{n-d}-1}\right).
\]
\end{corollary}
\begin{proof}
Following Equation \eqref{eq:union}, we have
\[
q^{2m}-1 = (q^m-1)|\H_1| + (q^m-1)|\H_2|.
\]
It is clear that the codewords of rank weight $d$ correspond to the hyperplanes in $\mathcal{H}_1$ and those of rank weight $n$ correspond to the hyperplanes in $\H_2$, where $\H_1, H_2$ are as described in Theorem \ref{pro:spread}. From Equation \eqref{eq:H1} it follows that the number of codewords of rank weight $d$ is equal to $(q^m-1)|\H_1| = (q^m-1)\frac{q^n-1}{q^{n-d}-1}$. Since the number of hyperplanes of $\Fqm^2$ is $|\H|=q^m+1$, and since $\H_1$ and $\H_2$ give a partition of $\H$, the number of codewords of rank weight $n$ is $(q^m-1)|\H_2| = (q^m-1)(q^m+1 -  \frac{q^n-1}{q^{n-d}-1})$. 
\qed\end{proof}

\begin{theorem}\label{decomposition}
Let $\C$ be an $[n,2,d]$ ATW rank metric code. Suppose that $n=l(n-d)$, then $\C$ is equivalent to a rank metric code with generator matrix $\GG=[\GG_1|\dots|\GG_l]$, where, for $1\leq i\leq l$, the columns of each $\GG_i$'s belong  to a hyperplane $H_i$ and the $H_i$'s are pairwise distinct.
\end{theorem}
\begin{proof}
Theorem \ref{pro:spread} says that $X$ admits an $(n-d)$-spread. Then we choose the columns of $\GG$ from the decomposition in Theorem \ref{thm:split}, where each block $\GG_i$ consists of a $\Fq$-basis of $H_i \cap X$ for $H_i \in \H_1$.
\qed\end{proof}

We have shown that ATW rank metric codes define $t$-spreads in the q-system $X$. The following theorem shows that these $t$-spreads satisfy a certain property.

\begin{theorem}\label{thm:6}
Let $\C$ be an $[n,2,d]$ rank metric code over $\Fqm/\Fq$ generated by $\GG$ with the corresponding $q$-system $X$. Let $(\Fqm^2,\D)$ be a Desarguesian $m$-spread. The following assertions are equivalent.
\begin{enumerate}[(i)]
    \item\label{thm:6-i} $\C$ is an ATW rank metric code.
    \item\label{thm:6-ii} $(\Fqm^2,\D)$ induces an $(n-d)$-spread $(X,\D_X)$ of $X$.
\end{enumerate}
\end{theorem}

\begin{proof}
\begin{itemize}
    \item[]\hspace{-0.8cm}$\eqref{thm:6-i} \Rightarrow \eqref{thm:6-ii}$: Since $\C$ is an ATW rank metric code, following Theorem \ref{pro:spread} we know that $\D_X = \{X \cap H : H \in \H_1\}$ is an $(n-d)$-spread of $X$. Also, from Definition \ref{def:subspread}, it follows that $(X,\D_X)$ is an $(n-d)$-subspread of $(\Fqm^2,\D)$.
    \item[]\hspace{0.2cm}$\eqref{thm:6-ii} \Rightarrow \eqref{thm:6-i}$: Suppose that $(X,\D_{X})$ is an $(n-d)$-subspread of $(\Fqm^2,\D)$. Let $\H$ be the set of all hyperplanes in $\Fqm^2$. Then from Definition \ref{def:desarguesian} and Definition \ref{def:subspread}, it follows that an element of $\H$ intersects $X$ in an $\Fq$-space of dimension either $n-d$ or $0$. From Lemma \ref{lem:2}, we have that for any $H \in \H$, $H = H_\cc$  for some $\cc \in \C$ and conversely, any $\cc\in \C$ defines a hyperplane $H_\cc$ (two linearly dependent codewords define the same hyperplane). Then Equation \eqref{eq:rank} in Lemma \ref{lem:2} implies that the only possible rank for any non-zero codewords are $d$ and $n$.
    \qed
\end{itemize}
\end{proof}

In the following result, we give a condition for any general $(n-d)$-spread of $\Fq^n$ to induce $[n,2,d]$ ATW rank metric codes.

\begin{corollary}
Let $n, m, d$ be positive integers with $d\leq n \leq m$ and let $\S$ be an $(n-d)$-spread of $\Fq^n$. Let $X \subseteq \Fqm^2$ be a $q$-system and suppose that $\GG$ is a $(2 \times n)$-matrix whose columns consist of an $\Fq$-basis of $X$. If for any $S\in \S$, $\GG S^T = \{\GG\xx^T\colon \xx\in S \}$ is contained in a hyperplane $H$ of $\Fqm^2$ and each different $S$ correspond to different $H$, then $\GG$ generates an ATW rank metric code.
\end{corollary}

\section{Constructions of ATW rank metric codes}\label{sec:4}
In Example \ref{exa:1} of Section \ref{sub:3.1} we have seen a construction of ATW rank metric codes. We construct ATW rank metric codes for more general parameters using suitable MRD codes. 

\begin{theorem}\label{thm:MRDto2weight}
Let $d,l,m, n$ be positive integers such that $d \leq n\leq m$, $(n-d) | m$, and $n = l(n-d)$. Suppose $\C$ is a non-degenerate $[l,2,l-1]$ MRD code over $\Fqm/\F_{q^{n-d}}$ with a generator matrix $\GG$. Fix a basis $\{a_1, \ldots, a_{n-d}\}$ of $\F_{q^{n-d}}$ over $\Fq$ and let $\tilde{\GG}$ be the $2 \times l(n-d)$ block matrix $(\tilde{\GG}_1| \ldots| \tilde{\GG}_l)$ where the $(n-d)$ columns of $\tilde{\GG}_i$ are given by $\{a_j \GG_i, j=1, 2, \ldots, n-d \}$ with $\GG_i$ being the $i$-th column of $\GG$. Then $\tilde{\GG}$ generates an $[n,2,d]$ ATW rank metric code $\tilde{\C}$ over $\Fqm/\Fq$.
\end{theorem}
\begin{proof}
We claim that if the codeword $c = (c_1, c_2) \GG$ of $\C$ has rank weight $l$ (respectively, $l-1$), then the codeword $\tilde{c} = (c_1, c_2) \tilde{\GG}$ of $\tilde{\C}$ has rank weight $l(n-d)$ (respectively, $(l-1)(n-d)$). 
It is easy to see that if the claim is proved then $\tilde{\C}$ is an ATW $[n,2,d]$ code over $\Fqm/\Fq$.

\textbf{Proof of the claim}: Let $X$ (resp. $\tilde{X}$) be the $q$-system corresponding to the generator matrix $\GG$ (resp. $\tilde{\GG}$). It is crucial to note that the $\Fq$-space $\tilde{X}$ and the $\F_{q^{n-d}}$-space $X$ are same as a set. Now let the codeword $c = (c_1, c_2) \GG$ of $\C$ has rank weight $l$. Thus from equation \eqref{eq:psi}, we get $\dim_{\F_{q^{n-d}}} H_c \cap X = 0$ and hence, $\dim_{\Fq} H_c \cap \tilde{X} = 0$. Therefore, $rank(\tilde{c}) = n - \dim_{\Fq} H_c \cap \tilde{X} = n$. Similarly, for the case when $rank(c)=l-1$ we have $\dim_{\F_{q^{n-d}}} H_c \cap X = 1$ (follows from equation \eqref{eq:psi}). Then $\dim_{\Fq} H_c \cap \tilde{X} = \dim_{\Fq} H_c \cap {X} = n-d$ and therefore, $rank(\tilde{c})=n -(n-d) = d$.
\qed\end{proof}

 The above theorem says that if $X$ is a $q^{n-d}$-system corresponding to an $[l,2,l-1]$ MRD code over $\Fqm/\F_{q^{n-d}}$, then $X$ is a $q$-system corresponding to an ATW $[n,2,d]$ code $\tilde{\C}$ over $\Fqm/\Fq$. We call the code $\tilde{\C}$ the \emph{ATW rank metric code induced} by the MRD rank metric code $\C$.
\begin{corollary}
 Let $d,l,m,n,$ be positive integers such that $d \leq n\leq m$, $n = l(n-d)$, and $(n-d) | m$. Let $X$ be a $q$-system corresponding to an $[n,2,d]$ ATW code $\tilde{\C}$ over $\Fqm/\Fq$. Then $\tilde{\C}$ is induced by an $[l,2,l-1]$ MRD code $\C$ over $\Fqm/\F_{q^{n-d}}$ if and only if $X$ is also an $\F_{q^{n-d}}$-space.
 \end{corollary} 
It is natural to ask if all ATW rank metric codes are induced by MRD codes as shown in the previous construction. The answer is affirmative when $d=n/2$.

\begin{theorem}\label{thm:3}
Let $\C$ be an $[n,2,d=n/2]$ rank metric code over $\Fqm/\Fq$. Then $\C$ is an ATW rank metric code if and only if $d$ divides $m$ and $\C$ is equivalent to a code with generator matrix 
\[
\begin{pmatrix}
0 & 0  & \dots & 0 & 1 & \alpha_2 & \dots & \alpha_d \\
1 & \alpha_2 & \dots & \alpha_d & 0 & 0  & \dots & 0
\end{pmatrix}, \text{ such that } \<1,\alpha_2,\dots,\alpha_d\>_{\Fq}=\mathbb{F}_{q^d}.
\]

\end{theorem}
\begin{proof}
 Following Theorem \ref{decomposition}, we know that $\C$ has a generator matrix of the form $\GG=[\GG_1|\GG_2]$ where the columns of $\GG_1$ (resp. $\GG_2$) belong to a hyperplane $H_1$ (resp. $H_2$). For $i=1,2$, let $\xx_i\in \Fqm^2$ such that $H_i = \mathrm{Ker}~\psi_{\cc_i}$ where $\cc_i = \xx_i\GG$ as in Equation \eqref{eq:psi}. Then $\C$ has another generator matrix 
$
\GG' = \begin{pmatrix}
\xx_1 \\
\xx_2
\end{pmatrix} \GG = 
\begin{pmatrix}
\0 & \ee_1 \\
\ee_2 & \0
\end{pmatrix},
$
where  $\ee_1$ and $\ee_2$ both have rank $d$. 

Without loss of generality we can assume that $\ee_1=(1,\alpha_2\dots,\alpha_d)$ and $\ee_2=(1,\beta_2,\dots,\beta_d)$ where $\alpha_i, \beta_i \in \Fqm^{\times}$. Take the codeword $(\ee_2|\ee_1)$, because the first and second part both contain $1$, then $0<\rank(\ee_2|\ee_1)<n$ and therefore, $\rank(\ee_2|\ee_1)=d$. Thus the $\beta_i$'s are $\Fq$-linear combination of the $\alpha_i's$ and $1$. Hence, $\C$ is equivalent to a code with generator matrix $\overline{\GG}$ of the form
$
\overline{\GG}=\begin{pmatrix}
\0 & \ee_1 \\
\ee_1 & \0
\end{pmatrix}.
$
For $i=2,\dots,d$, let $\cc^{(i)}=(\ee_1|\alpha_i\ee_1)$. Because the first and second part of $\cc^{(i)}$ contain $\alpha_i$, then $0<\rank(\cc^{(i)})<n$ and thus $\rank(\cc^{(i)}) = d$. This implies that $\alpha_i\alpha_j\in \<1,\alpha_2,\dots,\alpha_d\>_{\Fq}$ for all $2\leq i,j\leq d$. So the space $\<1,\alpha_2,\dots,\alpha_d\>_{\Fq}$ is in fact a subfield, say $K$, of $\Fqm$ with $[K:\Fq]=d$ and hence $d \mid m$.
In this case, $\C$ is a $[2d,2,d]$ rank metric code equivalent to a code with generator matrix 
\[
\GG" = \begin{pmatrix}
0 & 0  & \dots & 0 & 1 & \alpha_2 & \dots & \alpha_d \\
1 & \alpha_2 & \dots & \alpha_d & 0 & 0  & \dots & 0
\end{pmatrix}, \text{ such that } \<1,\alpha_2,\dots,\alpha_d\>_{\Fq}=\mathbb{F}_{q^d}. 
\]
For the converse, it follows from the arguments in Example \ref{exa:1} that $\GG"$ generates a non-degenerate ATW rank metric codes.
\qed\end{proof}

For $d=n/2$, we showed that the ATW rank metric code essentially corresponds to a Desarguesian spread $(\mathbb{F}_{q^{n/2}},\Delta)$. This gives an interesting property of Desarguesian subspreads.

\begin{corollary}\label{cor:desarguesiansubspread}
Let $(\Fqm^2,\D)$ be a Desarguesian spread. Let $X\subseteq \Fqm^2$ be an $\Fq$-space of dimension $n$ and suppose that $(\Fqm^2,\D)$ induces a subspread i.e. $n/2$-subspread, $(X,\D_{X})$ of $X$, then $(X,\D_{X})$ is a Desarguesian spread.
\end{corollary}
\begin{proof}
Since $(X,\D_{X})$ is a subspread, then by definition, the elements of $\D_{X}$ has dimension $n/2$. Since $(X, \D_{X})$ is a subspread induced by a Desarguesian spread, following Theorem \ref{thm:6}, $X$ defines an $[n,2,n/2]$ ATW rank metric code. Now Theorem \ref{thm:3} gives a generator matrix of the ATW code, from which it is clear that each element $S$ of $\D_X$ is an $\mathbb{F}_{q^{n/2}}$ vector space. So the induced subspread $(X,\D_X)$ is a Desarguesian spread of $X$. 
\qed\end{proof}

\section{Two-weight Hamming metric codes from two-weight rank metric codes}\label{sec:5}

As mentioned in the Introduction, there are already existing constructions of two-weight Hamming metric codes. In this section, as a small appendix to this work, we provide constructions of other two-weight Hamming metric codes via the expansion of rank metric codes as presented in \cite{ABNR21}. 
 To do this, it is better to also use the geometric approach for Hamming metric codes as introduced in \cite{TV91}. We only consider \emph{non-degenerate} Hamming metric codes i.e. no generator matrix of the code has a zero column.
 
\begin{definition}[Projective system]
Let $\PP=\PP(\Fq^k)$ be a projective space of dimension $k-1$ over a finite field $\Fq$. A projective $[n,k]_q$ system is a finite unordered family $\P$ of points of $\PP$, which does not lie in a (projective) hyperplane and such that $|\P|=n$ (counted with multiplicity).
\end{definition}
 
\begin{definition}[Projective system associated to a code]
Let $\C$ be an $[n,k]$ non-degenerate Hamming metric code over $\Fq$ and let $\GG\in \Fq^{k\times n}$ be a generator matrix of $\C$. A projective system $\P(\C)$ associated to $\C$ is the projective $[n,k]_q$ system in $\PP(\Fq^k)$ consisting of the columns of $\GG$ (with multiplicity).
\end{definition}

On the other hand, given a projective system $\P$ in $\PP(\Fq^k)$ of size $n$, one can consider a $k \times n$ matrix $\GG$ with the elements in the projective system $\P$ as its columns. This $k \times n$ matrix $\GG$ generates an $[n,k]$ code $\C(\P)$ over $\Fq$. In \cite[Chapter 1.1]{TV91}, Tsfasman and Vladut showed that, up to some notion of equivalences, there is a one-to-one correspondence between linear codes with Hamming metric and projective systems.

As it was shown in \cite{TV91}, if $\cc=\xx_\cc \GG$ is a codeword of $\C$, then the Hamming weight of $\cc$ is given by
 \[
\hw(\cc)= n-|\P\cap \H_\cc|,
\]
where the cardinality on the right-hand side is counted with multiplicity and $\H_\cc=\mathrm{Ker}~ \psi_{\xx_\cc}$ for $\psi_{\xx_\cc}$ as defined in Equation \eqref{eq:psi}. This is in fact the analogue of Lemma \ref{lem:2} and it is the key to the construction of Hamming metric codes from rank metric codes which we explain next.

\begin{definition}\label{def:12}
Let $X\subset \Fqm^k$ be a $q$-system and let $\sim$ be the equivalence relation on $X$ given by $\xx\sim \yy$ for $\xx,\yy\in X$ if and only if $\xx=\alpha\yy$ for some $\alpha\in \Fq^\times$. The projective system associated to $X$ is the set $\P(X)=(X/\sim) \backslash\{\0\}\subset \PP(\Fqm^k)$. 
\end{definition}

It follows easily from Definition \ref{def:12} that given an $[n,k]$ $q$-system $X$ in $\Fqm^k$, the associated projective system $\P(X)$ is an $\left[\frac{q^n-1}{q-1},k\right]$ projective system in $\PP(\Fqm^k)$.

 \begin{definition}\cite{ABNR21}
 Let $\C$ be an $[n,k]$ rank metric code over $\Fqm/\Fq$ corresponding to a $q$-system $X\subset \Fqm^k$ of dimension $n$ over $\Fq$.  A Hamming metric code $\fH(\C)$ associated to $\C$ is a $\left[\frac{q^n-1}{q-1},k\right]$ Hamming metric code over $\Fqm$ corresponding to the projective system $\P(X)$.
 \end{definition}
 
 For a rank metric code $\C$, the Hamming weights of the codewords in an associated Hamming metric code $\fH(\C)$ are determined by the rank weights of the codewords of $\C$.
 
 \begin{lemma}\label{weight-corr}\cite[Sec 4.2]{ABNR21}
 Let $\C$ be an $[n,k]$ rank metric code over $\Fqm/\Fq$ corresponding to a $q$-system $X\subset \Fqm^k$ with generator matrix $\GG\in \Fqm^{k\times n}$. Suppose that $\fH(\C)$ is a Hamming metric code associated to $\C$ with generator matrix $\GG' \in \Fqm^{k\times \left(\frac{q^n-1}{q-1}\right)}$ and let $\xx\in \Fqm^k$. If $\rank(\xx\GG) = t$, then $\hw(\xx\GG')=\frac{q^n-q^{n-t}}{q-1}$.
 \end{lemma}
From Lemma \ref{weight-corr}, a rank weight in the rank metric code $\C$ corresponds to a unique Hamming weight in the associated Hamming metric code $\fH(\C)$. Therefore we have the following theorem.

\begin{theorem}
Let $\C$ be an $[n,k,d]$ rank metric code over the extension $\Fqm/\Fq$. If $\C$ is a two-weight rank metric code, then the associated Hamming code $\fH(\C)$ is a $\left[\frac{q^n-1}{q-1},k,\frac{q^n-q^{n-d}}{q-1}\right]$ two-weight Hamming metric code. Moreover, if $\C$ is an antipodal two-weight rank metric code, then $\fH(\C)$ is an antipodal two-weight Hamming metric code.
\end{theorem}
\begin{proof}
Following Lemma \ref{weight-corr}, a codeword in $\C$ of rank $t$ corresponds to a codeword of $\fH(\C)$ with Hamming weight $\frac{q^n-q^{n-t}}{q-1}$. Hence if the weight distribution of $\C$ is $\{0,d,d'\}$, then the weight distribution of $\fH(\C)$ is $\left\lbrace 0,\frac{q^n-q^{n-d}}{q-1},\frac{q^n-q^{n-d'}}{q-1} \right\rbrace$.
\end{proof}

As a consequence of the previous theorem, Example \ref{exa:1} and Theorem \ref{thm:MRDto2weight} provide constructions of antipodal two-weight Hamming metric codes. Furthermore, Example \ref{exa:2} provides constructions of non-antipodal two-weight Hamming metric codes. 
\section*{Open questions}

We end this paper with some open questions. 
Theorem \ref{thm:MRDto2weight} gives a construction of $[n,2,d]$ ATW rank metric codes over $\Fqm /\Fq$ from suitable MRD codes over $\Fqm/\F_{q^{n-d}}$. In fact, for the case $d = n/2$ this is the only way of constructing ATW rank metric codes as proved in Theorem \ref{thm:3}. So it is natural to ask if this happens for the case when $d > n/2$? In other words, are there any $[n,2,d]$ ATW rank metric code which are not induced by MRD codes? Equivalently, in terms of $t$-spreads, one may ask if a $t$-subspread of a Desarguesian spread is again a Desarguesian $t$-spread. In other words, is there a non-Desarguesian t-subspread of a Desarguesian spread?

The main focus of this work was the study of antipodal two-weight rank metric codes. However, one may also want to explore the classification and characterization of non-antipodal two-weight rank metric codes.

%
%
%
%

\end{document}